\documentclass{scrartcl}

\usepackage{amsmath}
\usepackage{xargs}
\usepackage{ifthen}
\usepackage{dsfont}
\usepackage[disable]{todonotes}
\usepackage{amssymb}
\usepackage{bm} 
\usepackage{stmaryrd}
\usepackage{pgfplots}
\usepackage{subfig}
\usepackage{listings}
\usepackage{algorithm}
\usepackage{algpseudocode}
\usepackage{hyperref}

\newcommand{\N}{\mathds{N}}
\newcommand{\R}{\mathds{R}}

\newcommand{\UF}{\mathcal{U}}
\newcommand{\F}{\mathfrak{F}}

\newcommand{\eps}{\varepsilon}

\newcommand{\set}[1]{\left\{#1\right\}}
\newcommand{\abs}[1]{\left|#1\right|}
\newcommand{\norm}[1]{\left\|#1\right\|}
\newcommand{\To}{\mathop{\rightarrow}}
\newcommandx{\E}[2][2=\empty]{\ifthenelse{\equal{#2}{\empty}}{\mathrm{E}}{\mathrm{E}_{#2}}\!\left(#1\right)}

\newcommand{\Prob}[1]{\mathrm{Pr}\!\left(#1\right)}

\newcommand{\psfragsize}{\small}

\newcommandx{\psf}[3][3=c]{\psfrag{#1}[#3][]{$\psfragsize{#2}$}}
\newcommandx{\psft}[3][3=c]{\psfrag{#1}[#3][]{\psfragsize{#2}}}

\renewcommand{\vec}{\bm}
\newcommand{\supp}{\text{supp}}

\newcommand{\argmin}{\mathop{\text{argmin}}}
\newcommand{\argmax}{\mathop{\text{argmax}}}

\newtheorem{thm}{Theorem}[section]{\bfseries}{\rmfamily}
\newtheorem{lem}[thm]{Lemma}{\bfseries}{\rmfamily}
{\bfseries}{\rmfamily}
\newtheorem{cor}[thm]{Corollary}{\bfseries}{\rmfamily}
\newtheorem{defn}[thm]{Definition}{\bfseries}{\rmfamily}
\newtheorem{exa}[thm]{Example}{\bfseries}{\rmfamily}
\newtheorem{rem}[thm]{Remark}{\itshape}{\rmfamily}

{\bfseries}{\rmfamily}
{\bfseries}{\rmfamily}

\newenvironmentx{proof}[1][1=\empty]{\ifthenelse{\equal{#1}{\empty}}{\emph{Proof.}~}{\emph{Proof (#1).~}}}{\hfill$\square$\\[0.1\baselineskip]}

\definecolor{mygreen}{RGB}{28,172,0} 
\definecolor{mylilas}{RGB}{170,55,241}
\title{On Game-Theoretic Risk Management (Part Two)}
\subtitle{Algorithms to Compute Nash-Equilibria in Games with Distributions as Payoffs}
\author{Stefan Rass
\thanks{Universit\"{a}t Klagenfurt, Institute of Applied Informatics,
System Security Group, Universit\"{a}tsstrasse 65-67, 9020 Klagenfurt, Austria.
This work has been done in the course of consultancy for the EU Project
HyRiM (Hybrid Risk Management for Utility Networks; see https://hyrim.net),
led by the \emph{Austrian
Institute of Technology} (AIT; www.ait.ac.at). See the acknowledgement section.}\\
\texttt{stefan.rass@aau.at}}

\begin{document}

\maketitle

\begin{abstract}
\begin{center}
\textbf{Abstract}
\end{center}
The game-theoretic risk management framework put forth in the precursor
work ``Towards a Theory of Games with Payoffs that are
Probability-Distributions'' (\href{http://arxiv.org/abs/1506.07368}{ 	
arXiv:1506.07368 [q-fin.EC]}) is herein extended by algorithmic details on
how to compute equilibria in games where the payoffs are probability
distributions. Our approach is ``data driven'' in the sense that we assume
empirical data (measurements, simulation, etc.) to be available that can be
compiled into distribution models, which are suitable for efficient
decisions about preferences, and setting up and solving games using these
as payoffs. While preferences among distributions turn out to be quite
simple if nonparametric methods (kernel density estimates) are used,
computing Nash-equilibria in games using such models is discovered as
inefficient (if not impossible). In fact, we give a counterexample in which
fictitious play fails to converge for the (specifically unfortunate) choice
of payoff distributions in the game, and introduce a suitable tail
approximation of the payoff densities to tackle the issue. The overall
procedure is essentially a modified version of fictitious play, and is
herein described for standard and multicriteria games, to iteratively
deliver an (approximate) Nash-equilibrium. An exact method using linear
programming is also given.
\end{abstract}

\pagebreak

\tableofcontents \newpage
\section{Introduction}

Having laid the theoretical foundations in part one of this research (see
\cite{Rass2015b}), we now carry on describing the algorithmic aspects and
implementation notes for computing risk assurances on concrete games with
distributions as payoffs.

We start with a discussion on how to construct distributions from data in a
way that is suitable for efficiently deciding $\preceq$-preferences among the
empirical distribution estimates. In a nutshell, the $\preceq$-relation as
defined in \cite{Rass2015b} is as follows: assume that $R_1, R_2$ describe
the losses as bounded quantities between $1$ and some finite maximum (bounded
support), and let them have absolutely continuous measures w.r.t. the
Lebesgue or counting measure (assumption 1.3 in \cite{Rass2015b}).
\begin{defn}[Preference Relation over Probability Distributions]\label{def:preference}
Let $R_1,R_2$ be two bounded random variables $\geq 1$, whose distributions
are $F_1, F_2$. Write $m_{R_i}(k)$ for the $k$-th moment of $R_i$. We
\emph{prefer} $F_1$ \emph{over} $F_2$, written as
\begin{equation}\label{eqn:partial-ordering}
    F_1\preceq F_2:\iff \exists K\in\N\text{ s.t. }\forall k\geq K: m_{R_1}(k)\leq m_{R_2}(k)
\end{equation}
\emph{Strict preference} of $F_1$ over $F_2$ is denoted and defined as
\[
    F_1\prec F_2:\iff \exists K\in\N\text{ s.t. }\forall k\geq K: m_{R_1}(k)<m_{R_2}(k)
\]
Likewise, we define $F_1\equiv F_2\iff (F_1\preceq F_2)\land (F_2\preceq
F_1)$.
\end{defn}
Observe that $\equiv$ does not mean an identity (in the sense of equality)
between $F_1$ and $F_2$.

Concrete algorithms are given for estimating continuous distribution models
(section \ref{sec:payoff-estimation}), and on how to compare the following
pairs:
\begin{itemize}
  \item two continuous distributions with finite support (section
      \ref{sec:distribution-comparisons}) or infinite support (section
      \ref{sec:treating-infinite-supports}).
  \item continuous distribution vs. crisp number (section
      \ref{sec:det-vs-random}).
\end{itemize}

Chapter \ref{sec:computing-mgss} is devoted to algorithms for computing
multi-goal security strategies. It opens with a discussion on how to carry
over conventional fictitious play to $^*\R$ for one security goal (section
\ref{sec:one-dimensional-algorithms}), highlighting several nontrivial
pitfalls that must be avoided in a practical implementation. The full
algorithm is developed along a sequence of subsections, culminating in the
final description for the one-dimensional case in section
\ref{sec:one-dimensional-algorithms}. The generalization of the algorithm to
multicriteria distributions is derived on these grounds in section
\ref{sec:general-mgss-fp}.

The second major aspect of this work is computing security strategies in
multi-criteria games. These can be shown to correspond to Nash-equilibria in
properly transformed games, however, their computation is somewhat more
involved than in the real-valued case. For convenience of the reader, we
review the definition of multi-goal security strategies in section
\ref{sec:computing-mgss}, after having highlighted the practical and
theoretical obstacles in computing (general) Nash-equilibria in the special
kind of games that we consider. This discussion is subject of sections
\ref{sec:convergence-problems} and \ref{sec:numerics}.

\section{Comparing Distributions
Efficiently}\label{sec:distribution-comparisons} This section is devoted to
special cases of distribution models and how to compare them. In many cases,
we can avoid computing moment sequences, such as when the distribution can be
approximated or has compact support. The latter can be assured in kernel
density estimations using the Epanechnikov kernel, which is the first case
discussed now.

\subsection{Estimating Payoff Distributions from
Simulations}\label{sec:payoff-estimation} In light of assuming bounded
supports for all payoff distributions (see the introduction or
\cite{Rass2014}), it is useful to estimate payoff distributions from data
sets in a way that preserves compactness of the support and continuity of the
resulting distribution. To this end, we can construct a standard kernel
density estimator upon a kernel with compact support, such as the
Epanechnikov kernel
\begin{equation}\label{eqn:epanechnikov}
  k(x) := \left\{
            \begin{array}{ll}
              \frac 3 4(1-x^2), & \abs{x}\leq 1 \\
              0, & \hbox{otherwise,}
            \end{array}
          \right.
\end{equation}
plotted in figure \ref{fig:epanechnikov}.

\begin{figure}[b!]
  \centering
%
%
\includegraphics[scale=0.7]{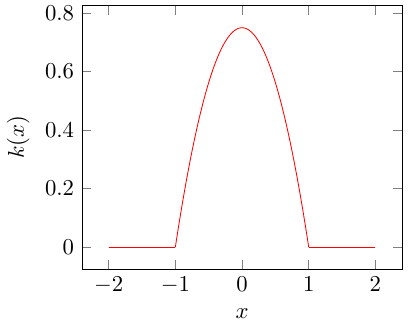}
  \caption{Plot of the Epanechnikov kernel}\label{fig:epanechnikov}
\end{figure}

Let $N=\set{x_1,\ldots,x_n}\subset\R$ be a sample of $n$ real-valued
simulation results, all of which have been harvested under a fixed (i.e.,
constant) configuration of choices (for example, strategies of all players in
the game under consideration).

Nadaraya's theorem then assures uniform convergence of the kernel density
estimator
\begin{equation}\label{eqn:kernel-density-estimator}
  \hat f(x) = \frac 1{n\cdot h}\sum_{i=1}^n k\left(\frac{x_i-x} h\right),
\end{equation}
towards the unknown payoff distribution density $dF/dx$, provided that the
latter is uniformly continuous, and the bandwidth parameter $h$ is set to
$h:=c\cdot n^{-\alpha}$ for constants $c>0$ and $0<\alpha<\frac 1 2$.
Practical heuristics (rules of thumb) are available and implemented in
various software packages, like \texttt{R} \cite{RDevelopmentCoreTeam2011}.

While this choice is convenient for evaluating the $\preceq$-relation,
computing equilibria requires a different choice for the kernel function,
which will be the Gaussian kernel (or any kernel with infinite support). We
go into details later in section \ref{sec:convergence-problems} and
afterwards.

\subsection{Comparing Kernel Density
Models}\label{sec:kde-payoff-estimation} Given two kernel density estimates
by using their hyperreal representatives is easy, and made precise in the
following procedure (cf. \cite{Rass2015d}):

Let $\hat f_1,\hat f_2$ be two kernel density estimates
\eqref{eqn:kernel-density-estimator}, constructed from sample sets of sizes
$n_1, n_2$, with respective bandwidths $h_1, h_2$. Furthermore, let us assume
that the kernel has compact support, which is -- for example -- automatically
satisfied for the Epanechnikov kernel. Let the samples be in ascending order
(that can be arbitrary for ties), i.e., the densities $\hat f_1,\hat f_2$ are
constructed from the values $x_1\leq x_2\leq \cdots\leq x_{n_1}$ and $y_1\leq
y_2\leq \cdots\leq y_{n_2}$. To decide which one of $\hat f_1,\hat f_2$ is
preferred over the other, it suffices to look at the largest value in each
sample. Here, we distinguish the following cases:
\begin{itemize}
  \item If $x_{n_1}+h_1<y_{n_2}+h_2$, then $\hat f_1\prec \hat f_2$.
  \item If $y_{n_2}+h_2<x_{n_1}+h_1$, then $\hat f_2\prec \hat f_1$.
  \item Otherwise, we must have $x_{n_1}+h_1=y_{n_2}+h_2$, in which case we
      go on comparing as follows:
      \begin{itemize}
        \item If $x_{n_1}<y_{n_2}$, then $\hat f_1\prec \hat f_2$
        \item If $y_{n_2}<x_{n_1}$, then $\hat f_2\prec \hat f_1$
        \item Otherwise, we must have $x_{n_1}=y_{n_2}$, and therefore
            also $h_1=h_2$. In that case, we discard the pair
            $x_{n_1},y_{n_2}$ and repeat the comparison procedure on the
            next samples $x_{n_1-1}, y_{n_2-1}$ in decreasing order on
            both lists.
      \end{itemize}
\end{itemize}
The effort for sorting then makes the above procedure decide the
$\prec$-relation with complexity $O(n\log n)$, where $n=\max\set{n_1,n_2}$.

The correctness of this method is immediately evident upon the fact that any
two distinct points $z_i\neq z_j$ in a sequence of samples contribute to the
(respective) density using bandwidth $h$ on a range $[z-h,z+h]$, where
$z\in\set{z_i,z_j}$. A subtle issue arises in the case of categorial data
(say, if the outcome is rated as ``low'', ``medium'', ``high''), for which
identical samples may accumulate at identical positions. In this case, the
density with less samples in the higher range will be preferred. In the above
process, this will cause identical samples to be removed until either\ldots
\begin{itemize}
  \item \ldots exactly one of the distributions has no further samples at
      position $z$, which makes this preferred (since the other
      distribution assigns a nonzero likelihood to larger damage
      possibilities), or
  \item \ldots both distributions have an equal amount of probability mass
      on $z$, in which case the respective density functions are identical
      and the difference between the densities is zero.
\end{itemize}

In any case, this means that we can ultimately assume $z_i\neq z_j$, so that
the respective intervals are not congruent. W.l.o.g., assume $z_i<z_j$, then
the density is strictly positive on the range $[z_j-h,z_j+h]\setminus
[z_i-h,z_i+h]$, whereas it vanishes outside the interval. Now, apply this
reasoning to two densities, having their shapes on the right end of their
support being defined by the kernels centered around the maximum value in
either data sample. Depending on the bandwidth in each estimator and the
location of the interval, it remains to determine which density reaches
farther out and remains positive when the other density vanishes. The
arguments in the proof of \cite[lemma 2.4]{Rass2015b} exhibit that the
distribution whose support strictly overlaps the other that will have its
moments grow faster than the distribution that it is compared to. Hence, the
preference relation $\prec$ can be decided upon checking which density
estimate has the longer tail.

To improve flexibility in this regard, let us look at mixture distribution
models, which are distribution functions whose density is a mix of several
(simpler) density functions, i.e.,
\[
    f_{\text{mix}}(x)=\sum_{i=1}^n\omega_i f_i(x),\quad\text{where}\quad \sum_{i=1}^n\omega_i=1.
\]
From the mix of densities, the mixed distribution is instantly discovered to
be
\[
    F_\text{mix}(x) = \sum_{i=1}^n\omega_i F_i(x)
\]
Likewise, the moments of $X\sim F_{\text{mix}}$ satisfy
\begin{equation}\label{eqn:moments-of-mixes}
    \E{X^k}=\int_0^\infty x^kf_{\text{mix}}(x)dx=\sum_{i=1}^n\omega_i \E{X_i^k},\quad\text{where}\quad X_i\sim F_i
\end{equation}
From the last identity, we instantly obtain that for any two distributions
$F_1\preceq F_2$ and $G_1\preceq G_2$, preference holds for the mix in the
same way, i.e. (by induction), if $F_i\preceq G_i$ for $i=1,2,\ldots,n$, then
$F_{\text{mix}}\preceq G_{\text{mix}}$, with the mixtures defined as above.

Mixture distribution models are particularly handy as they can be shown to
approximate \emph{any} distribution up to arbitrary precision\todo{Zitiere
Bayesian Choice}.

Having mixtures as models, it is interesting to see what happens if the
mixture is such that components compare alternatingly. In other words, let
$F_1=\omega_1 G_1+\omega_2 G_2$ and $F_2 = \omega'_1 H_1 + \omega'_2 H_2$,
where $G_1\preceq H_1$ but $G_2\succeq H_2$.

Here, we can take advantage of the hyperreal representation of distributions
which allows us to do arithmetics ``as usual'', whereas it must strictly be
emphasized that the result of any such arithmetics no longer represents a
valid distribution (nor is the arithmetic done here in any way statistically
meaningful; it is a mere vehicle towards the sought conclusion).

More specifically, recall that any of the above distributions is represented
by a sequence of moments, and the comparison is based on which sequence grows
faster than the other.

This means that by virtue of \eqref{eqn:moments-of-mixes} and by letting the
sequences $(g_{1,n})_{n\in\N}$, $(g_{2,n})_{n\in\N}$, $(h_{1,n})_{n\in\N}$
and $(h_{2,n})_{n\in\N}$ represent the distributions $G_1,G_2,H_1$ and $H_2$,
we end up looking at the limit
\[
    a = \lim_{n\to\infty}\big[\underbrace{(\omega_1 g_{1,n}+\omega_2 g_{2,n})}_{\text{represents } F_1}-\underbrace{(\omega_1' h_{1,n}+\omega_2' h_{2,n})}_{\text{represents }F_2}\big]
\]
If $a\leq 0$ then $F_1\preceq F_2$, with strict preference if $a<0$.
Otherwise, we have $F_1\succeq F_2$, likewise with strict preference. More
detailed criteria call for additional hypotheses, such as the existence of
closed-form expressions for the involved moments. Many special cases,
however, are easy to decide, such as mixes of distributions with compact
support. In the case of finite mixtures, $F_{\text{mix}}$ is supported on
$\bigcup_{i=1}^n \supp(F_i)$, and has itself compact support. Then, all the
previously given criteria apply. Hence, the above limit expression comes into
play when comparing distributions with infinite support, and calls for a
full-fledged analysis only if approximations (by truncating the
distributions) are not reasonable.

\subsection{Comparing Deterministic to Random
Variables}\label{sec:det-vs-random} Let $a\in\R$ be a deterministic outcome
of an action, and let $Y\sim F$ be a random variable with non-degenerate
distribution supported on $\Omega=[0,b]$ whose density is $f$. We know that
$a$ yields a moment sequence $(a^k)_{k\in\N}$. A comparison to $Y$ is easy in
every of the three possible cases:
\begin{enumerate}
  \item If $a<b$, then $a\preceq Y$: to see this, choose $\eps<(b-a)/3$ so
      that $f$ is strictly positive on a compact set $[b-\eps,b-2\eps]$
      (note that such a set must exist as $f$ is continuous and the support
      ranges until $b$). We can lower-bound the $k$-th moment of $Y$ as
      \begin{align*}
        \int_0^b y^kf(y)dy &\geq \left(\inf_{[b-2\eps,b-\eps]} f\right)\cdot \int_{b-2\eps}^{b-\eps} y^kdy\\
        & = \frac
1{k+1}\left[(b-\eps)^{k+1}-(b-2\eps)^{k+1}\right].
      \end{align*}
      Note that the infimum is positive as $f$ is strictly positive on the
      compact set $[b-2\eps,b-\eps]$. The lower bound is essentially an
      exponential function to a base larger than $a$, since $b-2\eps>a$,
  and thus (ultimately) grows faster than $a^k$.
  \item If $a>b$, then $Y\preceq a$, since $Y$ -- in any possible
      realization -- leads to strictly less damage than $a$. The formal
      argument is now based on an upper bound to the moments, which can be
      derived as follows:
      \begin{align*}
      \int_0^b y^k f(y)dy & \leq (\sup_{[0,b]} f)\cdot \int_0^b y^kdy = (\sup_{[0,b]} f)\frac 1{k+1}b^{k+1}.
      \end{align*}
      It is easy to see that for $k\to\infty$, this function grows slower
      than $a^k$ as $a>b$, which leads to the comparison result.
  \item If $a=b$, then we apply the mean-value theorem to the integral
      occurring in $\E{Y^k}=\int_0^a y^kf(y)dy$ to obtain an $\xi\in[0,a]$
      for which
      \[
        \E{Y^k}=\xi^k\underbrace{\int_0^af(y)dy}_{=1}=\xi^k\leq a^k
      \]
      for all $k$. Hence, $Y\preceq a$ in that case. An intuitive
      explanation is obtained from the fact that $Y$ may assign positive
      likelihood to events with less damage as $a$, whereas a deterministic
      outcome is always larger or equal to anything that $Y$ can deliver.
\end{enumerate}

\subsection{Distributions with Infinite
Support}\label{sec:treating-infinite-supports} This case is difficult in
general, but simple special cases may occur, for example, if we compare a
distribution with compact support to one with infinite support (such as
extreme value distributions or ones with long or fat tails). Then, the
compactly supported distribution is always preferred, by the same argument as
used above (and in the proof of the invariance of $\preceq$ w.r.t. the
ultrafilter used to construct $^*\R$; see \cite[lemma 2.4]{Rass2015b}).

The unfortunate occasions are those where:
\begin{itemize}
  \item both distributions have infinite support, and
  \item neither \cite[Lemma 2.7]{Rass2015b} nor any direct criteria (such
      as eq. (13) in \cite{Rass2015b}) apply, and
  \item an approximation cannot be done (for whatever reason).
\end{itemize}
Then we have to work out the moment sequences explicitly. This situation is
indeed problematic, as without assuming bounded supports, we can neither
guarantee existence nor divergence of the two moment sequences.

Appropriate examples illustrating this problem can easily be constructed by
defining distributions with alternating moments from the representation by
the Taylor-series expansion of the characteristic function (see
\cite{Rass2015b} for an example).

As was explained in \cite{Rass2015b}, however, explained that extreme value
distributions strictly compare to one other depending depending on their
parameterizations, by applying lemma the criteria derived in
\cite{Rass2015b}. Mixes of such distributions can perhaps replace an
otherwise unhandy model.

\subsection{On Paradoxical Comparisons and Finding Good
Approximations}\label{sec:paradoxical-choices} In part one of this report,
finite approximations to distributions with infinite support were proposed to
fit these into our preference relation. It turns out, however, that even
distributions with finite support may lead to paradoxical and unexpected
results in terms of $\preceq$-preference. That is, we would in any case
prefer the distribution with smaller support, but this is not necessarily the
one giving us less damage in most of the cases.

To illustrate the problem, consider the two distributions plotted in figure
\ref{fig:quantile-based-approximations} (where the densities have been scaled
only to properly visualize the issue). Observe that $F_2$ assigns all its
mass around larger damage, while $F_1$ generally gives much lower damages
except for rare cases that exceed the events that can occur under $F_2$.
These rare cases, however, extend beyond the support of distribution $F_1$,
which based on the characterization by a sequence of would clearly let us
prefer $F_2$ over $F_1$ (Figure \ref{fig:paradoxical-preference}). Indeed, it
is easy to see that such a result is not what we would expect or want in
practice.

\begin{figure}[h!]
  \centering
  \subfloat[Distributions $F_2\preceq F_1$]{\label{fig:paradoxical-preference}\includegraphics[width=0.45\textwidth]{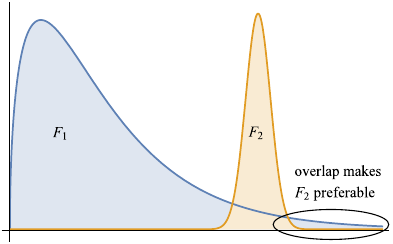}}\quad
  \subfloat[Truncated distributions $\hat F_1\preceq \hat F_2$]{\label{fig:corrected-preference}\includegraphics[width=0.45\textwidth]{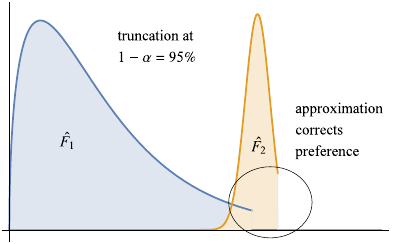}}\\
  \caption{Correcting Paradoxical Comparisons by Quantile-Based Approximations \cite{Rass2015d}}\label{fig:quantile-based-approximations}
\end{figure}

Truncating the distributions properly can, however, easily fix the issue. To
this end, let $\alpha>0$ be a \emph{chosen and fixed} threshold and cut off
each distribution as soon as it exceeds $1-\alpha$ of its mass. Precisely, we
could thus cut off a distribution at its $(1-\alpha)$-quantile, denoted as
$q_{1-\alpha}=F^{-1}(1-\alpha)$ (where $F^{-1}$ is the inverse of the
distribution function; called the \emph{quantile function}). Now, reconsider
the previous situation with $F_1$ and $F_2$, now being truncated properly
(Figure \ref{fig:corrected-preference}). Since both distributions' supports
extend only to the point when they have assigned 95\% of their mass, it is
instantly revealed that $F_2$ does so much earlier than $F_1$, hence making
$F_2$ clearly the preferred option here. The price of this fix is the
\emph{acceptance} of events that are less likely than $\alpha$ under $F_2$,
which can -- in a fraction of $\alpha$ among all possible cases -- give more
damage than $F_1$ could ever do.

In other words, common approximations based on a \emph{risk acceptance
threshold} $\alpha$, as described above, can avoid paradoxical preferences.
The assertion of theorem 2.14 in \cite{Rass2015b} would then fail in claiming
that extreme events are less likely under the preferred distribution. The
assertion, however, remains valid in a fraction of $(1-\alpha)\%$ of cases,
if we cut off the distributions.

\section{Computing Multi-Goal Security Strategies (MGSS)}\label{sec:computing-mgss}
Roughly speaking, the basic algorithm to compute MGSS in the sense of
\cite[Def.4.1]{Rass2015b} (generalizing precursor work in \cite{Rass2012}).
For the definition, let $\F\subseteq ^*\R$ be a subset of the hyperreals
$^*\R$ that represent probability distributions that are supported within a
compact subset of $[1,a]$ for some real number $a$. For two $n$-dimensional
vectors $\vec x=(x_1,\ldots,x_n),\vec y=(y_1,\ldots,y_n)$, we write $\vec
x\preceq_1 \vec y$, if an index $i$ exists for which $x_i\preceq y_i$
(irrespectively of what the other entries do).

\begin{defn}[Multi-Goal Security Strategy with Assurance]\label{def:MGSS}\hfill\\
A strategy $\vec p^*\in S_1$ in a two-person multi-criteria game with mixed
strategy spaces $S_1,S_2$. Let player 1 be a \emph{maximizer}, and have a
continuous payoff function $\vec u_1:S_1\times S_2\To\F^d$. A
\emph{Multi-Goal Security Strategy with Assurance} (MGSS) with
\emph{assurance} is a pair $(\vec v,\vec x^*)$, where $\vec
v=(V_1,\ldots,V_d)\in\F^d$ is vector of distributions and $\vec x^*\in S_1$,
if two criteria are met:
\begin{description}
  \item[Axiom 1: Assurance] The values in $\vec v$ are the component-wise
      guaranteed payoff for player 1, i.e. for all components $i$, we have
  \begin{equation}\label{eqn:ds-req0}
    V_i\preceq u_1^{(i)}(\vec p^*,\vec q)\qquad\forall\vec q\in S_2,
  \end{equation}
  with $\equiv$-equality being achieved by at least one choice $\vec q_i\in
  S_2$.
  \item[Axiom 2: Efficiency] At least one assurance becomes void if player
      1 deviates from $\vec p^*$ by playing $\vec p\neq \vec p^*$. In that
      case, some $\vec q_{\vec p}\in S_2$ exists (that depends on $\vec p$)
      such that
  \begin{equation}\label{eqn:ds-req2}
    \vec u_1(\vec p,\vec q_{\vec p})\preceq_1 \vec v.
  \end{equation}
\end{description}

\end{defn}

To practically compute MGSS, we actually ``simulate'' a gameplay in which
both players (honest defender and the full set of adversaries, each of which
corresponds to another security goal), record each others actions and
optimize their choices w.r.t. the empirical distributions\footnote{in the
literature, this is called a \emph{belief sequence}.}.

\subsection{The case of only one security
goal}\label{sec:one-dimensional-algorithms} We start with the simpler case in
which the MGSS is one-dimensional. This case boils down to an application of
regular fictitious play to a ``standard'' zero-sum game. The general case of
$d>1$ goals can easily be reduced to the one-dimensional case, as we will
show in section \ref{sec:general-mgss-fp}.

For the current matter, it suffices to decide upon one's best action, given
that the others choose their actions with known probabilities. This is
nothing else than a mix of distributions that we need to compare to our own
strategy (details of how this can be done have been discussed before).

To exactly specify a suitable fictitious play (FP) algorithm, we will
step-by-step transfer a MATLAB implementation of a regular FP algorithm as
described in \cite{Washburn2001} to our modified setting. Towards proving
correctness of the algorithm in our setting, the hyperreal numbers will turn
out useful.

\enlargethispage{\baselineskip} 
\begin{lstlisting}[caption=Fictitious play in
MATLAB for a zero-sum game with payoff matrix $\vec A\in\R^{n\times m}$ and
with a maximizing player 1/minimizing player 2,label=lst:matlab-fp]
% standard fictitious play in an two person 
% zero sum (mxn)-matrix game with 
% payoff matrix A. The computation delivers 
% an approximate Nash equilibrium
% over iterationCount steps (must be set in advance)
[m n] = size(A);
x = zeros(1, m);
y = zeros(1, n);
[vmin, r] = max(min(A')); 
[vmax, c] = min(max(A));
U = A(:, r);
V = zeros(1, n);
y(c) = y(c) + 1;
for i = 1:iterationCount-1
   [bestU, r]  =  max(U); % best response to player 2's actions
   x(r) = x(r) + 1;             % update player 1's behavior
   V = V + A(r, :);             % record player 1's payoff
   vmax = min(bestU / i, vmax); % update the upper bound to v
   [bestV, c]  =  min(V); % best response to player 1's actions
   y(c) = y(c) + 1;             % update player 2's behavior
   U = U + A(:, c);             % record player 2's payoff
   vmin = max(bestV / i, vmin); % update the lower bound to v
end
x = x / sum(x);   % absolute frequencies -> probabilities
y = y / sum(y);
v = x*A*y';       % average payoff (estimated)
\end{lstlisting}

Lines $6,\ldots, 13$ in the algorithm cover matters of initialization, where
we determine the size of the payoff structure ($n\times m$ being assumed
here), and initialize the solution vectors $\vec x$ and $\vec y$ to all zero.
Upon termination of the algorithm, $\vec x$ and $\vec y$ will approximate the
MGSS $\vec p^*$ and an optimal adversarial mixed strategy $\vec q^*$.

\begin{rem}
Note that our problem is to minimize player 1's loss, while player 2
(adversary) attempts to maximize the damage. This is essentially the opposite
of what is implemented in listing \ref{lst:matlab-fp}, so that our first
modification is to swap all calls to $\min$ and $\max$ in the entire
algorithm (lines $9\ldots 23$). The rest of the algorithm remains unchanged.
Nevertheless, we will continue the explanation letting player 1 be a
maximizer here, without loss of generality.
\end{rem}

Now for the for-loop (lines $14\ldots 23$): during the initialization phase,
the FP algorithm chooses an arbitrary initial pure strategy $1\leq i_1\leq n$
for player 1 (done in lines 9 and 10 of listing \ref{lst:matlab-fp}). From
then on, the players alternatingly choose their next pure strategies as a
best reply assuming that the other player selects at random from her/his
previous choices. That is, let the history of player 1's (row-player) moves
be $r_1,\ldots,r_k$ over the past $k$ iterations of FP, then player 2
(column-player) chooses his strategy $c_k$ so as to minimize the average
payoff
\begin{equation}\label{eqn:fp-average-payoff-1}
    c_k = \argmin_{j\in PS_2}\frac 1 k \sum_{\ell=1}^k a_{i_{\ell},j},
\end{equation}
when the payoff matrix is $\vec A = (a_{i,j})_{i,j=1}^{n,m}$. Let this choice
be $c_{k+1}$, then likewise, player 1 would in the next round choose
$r_{k+1}$ to maximize the average outcome
\begin{equation}\label{eqn:fp-average-payoff-2}
    r_{k+1}=\argmax_{i\in PS_1}\frac 1 k \sum_{\ell=1}^k a_{i,j_{\ell}}.
\end{equation}
These alternating choices are made in lines 15\ldots 17 and 19\ldots 21 in
the MATLAB code. Lines 16 and 20 simply count so-far played moves in a
vector, which, after normalization, will be the sought approximation to the
Nash-equilibrium $(\vec p^*,\vec q^*)$.

In light of the fact that our games pay the players in terms of entire
probability distributions rather than numeric values, line 17 and line 21
deserve a closer look: Indeed, the division by the iteration count in lines
18 and 22 is nothing else than the factor $1/k$ appearing in
\eqref{eqn:fp-average-payoff-1} and \eqref{eqn:fp-average-payoff-2}, so that
both expressions
\begin{equation}\label{eqn:fp-intermediate-mix}
\frac 1 k \sum_{\ell=1}^k F_{i_{\ell},j},\quad\text{and}\quad\frac 1 k \sum_{\ell=1}^k F_{i,j_{\ell}}
\end{equation}
yield valid distribution functions again, when we use our payoff distribution
matrix $\vec A=(F_{ij})$ instead of real values. As was previously noticed in
equation \eqref{eqn:moments-of-mixes} already, the corresponding hyperreal
representation of a weighted sum $\frac 1 k\sum_{\ell=1}^k F_{i_{\ell},j}$ is
identical to the correspondingly weighted sum of hyperreal representations of
each term (independently of the ultrafilter upon which the hyperreal space
$^*\R$ was constructed). In other words, working with distributions in
\eqref{eqn:fp-average-payoff-1} and \eqref{eqn:fp-average-payoff-2} amounts
to leaving the algorithm exactly unchanged, except for letting it operate on
hyperreal numbers instead of real numbers to represent the payoffs. By the
transfer principle \cite{Robinson1966}, convergence of the algorithm
identically holds in the space of hyperreals as it does in the real numbers
(since the respective statements are all in first-order logic\todo{genaue
Statements noch zitieren; eventuell den Satz von Losz}). It follows that the
resulting algorithm pointwise adds distribution or density functions in lines
17 and 21, and it divides by the iteration count as it would divide a
standard real number (only, it is a distribution function in our case). This
is permitted, since the pointwise scaling of a distribution function by a
real value amounts to scaling the respective hyperreal number (moment
sequence) element-wise by the same factor.

Convergence of the algorithm then implies that the approximations $\vec x$
and $\vec y$ will eventually approximate an equilibrium in the
$\preceq$-sense (asymptotically). Since the outcome of the algorithm over the
reals is an approximation to a Nash-equilibrium $(\vec x^*,\vec y^*)$ w.r.t.
the payoff being $\vec x^T\vec A\vec y$, the outcome of our algorithm is a
hyperreal pair $(\vec x^*,\vec y^*)$ that forms an equilibrium for the payoff
function $\vec x^T\vec A\vec y$, which, however, equals the sought overall
distribution of the damage $R$,
\begin{equation}\label{eqn:game-outcome-distribution}
    \Prob{R\leq r}=(F(\vec p,\vec q))(r)=\sum_{i,j}F_{ij}(r)\cdot p_i\cdot q_j,
\end{equation}
with independent choices by both players.

It is important to notice the difference of this outcome to that of standard
game-theory, where a ``repeated game'' means something essentially different
than here:
\begin{itemize}
  \item In standard game-theory, the players would choose their payoffs to
      maximize their long-run average over \emph{infinitely many
      repetitions} of the game.
  \item In our setting, as implied by
      \eqref{eqn:game-outcome-distribution}, the Nash-equilibrium optimizes
      the game's overall payoff distribution \emph{for a single round}.
      That is, this consideration is a single-shot, as opposed to
      conventional games, where we benefit only in the long-run from
      playing the equilibrium.
\end{itemize}
The whole trick is equation \eqref{eqn:game-outcome-distribution} resembling
a long-run average payoff for a conventional game, when the strategy choices
are made independently. Since our distributions boil down to mere hyperreal
\emph{numbers} in the space $\F\subset\R^\infty\slash\UF$, convergence there
is inherited from the convergence of regular FP over the reals. More
importantly, observe that the FP algorithm merely adds and scales
distributions (by a real valued factor), which is doable without any explicit
ultrafilter $\UF$. Thus, regular FP can be applied \emph{seemingly without
change} to our setting, except for the $\min$- and $\max$-functions being
computed w.r.t. $\preceq$-relation (and the criteria as specified above).
Alas, this impression will turn out to be practically flawed (example
\ref{exa:failure-of-fp}).

Let us now turn to the full pseudocode-specification of fictitious play in
our modified setting, which is algorithm \ref{alg:fp}. Note that this
algorithm is also adapted towards letting player 1 (the defender) be a
minimizer now.
\begin{algorithm}
\caption{Fictitious Play (with $\preceq$-minimizing first player)}\label{alg:fp}
\begin{algorithmic}[1]
\Require an $(n\times m)$-matrix $\vec A$ of payoff distributions $\vec
A=(F_{ij})$%
\Ensure an approximation $(\vec x,\vec y)$ of an equilibrium pair $(\vec
p^*,\vec q^*)$ and a two distributions $v_{low},v_{up}$ so that
$v_{low}\preceq F(\vec p^*,\vec q^*)\preceq v_{up}$. Here, $F(\vec
p^*,q^*)(r)=\Prob{R\leq r}=\sum_{i,j}F_{ij}(r)\cdot p_i^*q_j^*$.%
\State initialize $\vec x\gets \vec 0\in\R^n$, and $\vec y\gets \vec 0\in\R^m$%
\State $v_{low} \gets$ the $\preceq$-minimum over all
column-maxima%
\State $r \gets $the row index giving
$v_{low}$%
\State $v_{up} \gets$ the $\preceq$-maximum over all row-minima%
\State $c \gets$ the column index giving $v_{up}$%
\State $\vec u\gets (F_{1,c},\ldots,F_{n,c})$%
\State $y_{c}\gets y_{c} + 1$\Comment{$\vec y=(y_1,\ldots,y_m)$}%
\State $\vec v\gets \vec 0$\Comment{initialize $\vec v$ with $m$ functions that are zero everywhere}%
\For{$k=1,2,\ldots$}\label{lbl:fp-for-loop}%
    \State $u^*\gets$ the $\preceq$-minimum of $\vec u$%
    \State $r\gets$ the index of $u^*$ in $\vec u$%
    \State\label{lbl:vup-update} $v_{up} \gets$ the $\preceq$-maximum of $\set{u^*/k, v_{up}}$\Comment{pointwise scaling of the distribution $u^*$}%
    \State $\vec v\gets \vec
    v+(F_{r,1},\ldots,F_{r,m})$\Comment{pointwise addition of functions}%
    \State\label{lbl:fp-choice-counts} $x_{r}\gets x_{r}+1$\Comment{$\vec x=(x_1,\ldots,x_n)$}%
    \State\label{lbl:start-player2} $v_*\gets$ the $\preceq$-maximum of $\vec v$%
    \State $c\gets$ the index of $v_*$ in $\vec v$%
    \State\label{lbl:vlow-update} $v_{low}\gets$ the $\preceq$-minimum of $\set{v_*/k,
    v_{low}}$\Comment{pointwise scaling of the distribution $v_*$}%
    \State $\vec u\gets\vec u + (F_{1,c},\ldots,F_{n,c})$\Comment{pointwise addition of functions}%
    \State\label{lbl:end-player2} $y_{c}\gets y_{c}+1$\Comment{$\vec y=(y_1,\ldots,y_m)$}%
    \State\label{lbl:exit-criterion} exit the loop upon convergence\Comment{concrete
    condition given below}%
\EndFor%
\State\label{lbl:result-normalization} Normalize $\vec x, \vec y$ to unit
total sum\Comment{turn $\vec x,\vec
y$ into probability distributions.}%
\State\label{lbl:fp-finish}\Return{$\vec p^*\gets \vec x, \vec q^*\gets\vec
y,$ and $F(\vec p^*,\vec q^*)\gets \sum_{i,j}F_{ij}(r)\cdot x_i\cdot
y_j$}\Comment{$\approx (\vec p^*)^T\vec A\vec q^*$}
\end{algorithmic}
\end{algorithm}

Convergence of algorithm \ref{alg:fp} is inherited from the known convergence
results in the standard setting of fictitious play, as shown in listing
\ref{lst:matlab-fp}. Precisely, we have the following result, originally
proven by J.Robinson \cite{Robinson1951}. Our version here is an adapted
compilation of proposition 2.2 and proposition 2.3 in \cite[chapter
2]{Fudenberg1998}:
\begin{thm}\label{thm:standard-fp-convergence}
Under fictitious play in $^*\R$, the empirical distributions converge (along
an unboundedly growing sequence of hyperreal integers), if the game is
zero-sum.
\end{thm}


\begin{rem}
The zero-sum assumption in theorem \ref{thm:standard-fp-convergence} needs
some elaboration here: let $\vec A$ be the original payoff structure with all
distribution-valued entries. Let $\underline{\vec A}$ be its hyperreal
representation (by replacing each distribution by its representative moment
sequence). Then, a zero-sum game -- as considered in theorem
\ref{thm:standard-fp-convergence} -- would assign the payoffs
$(\underline{\vec A},-\underline{\vec A})$, in which player 2's structure no
longer corresponds to a valid matrix of probability distributions. However,
FP merely means the players choosing their moves as a \emph{best response} to
the so-far recorded history of strategies. For player 2, $\preceq$-maximizing
his payoff on $-\underline{\vec A}$ is the same as $\preceq$-minimizing the
payoff on $\underline{\vec A}$ by the properties of the ordering on the
hyperreals. Thus, his choices are indeed based on a valid probability
distribution matrix.
\end{rem}

We stress that theorem \ref{thm:standard-fp-convergence} talks about the
convergence of \emph{empirical distributions}, rather than the convergence of
the game's \emph{value} only (which is the somewhat weaker statement found in
\cite{Robinson1951}). 
The importance of this insight for us lies in the convergence criterion that
it delivers:
\begin{quote}
\underline{Convergence criterion for algorithm \ref{alg:fp}, line
\ref{lbl:exit-criterion}}: let $\vec x_k$ denote the empirical absolute
frequencies of strategy choices, as recorded by algorithm \ref{alg:fp} in
line \ref{lbl:fp-choice-counts}. Fix any $\eps>0$ and some vector-norm
$\norm{\cdot}$ on $\R^n$, and terminate the algorithm as soon as $\frac 1
k\norm{\vec x_{k+1}-\vec x_k}<\eps$.
\end{quote}
Alternative convergence criteria, say on the difference
$\abs{v_{up}-v_{low}}$ are also found in the literature, but are not
applicable to our setting here: to see why, recall that both, $v_{up}$ and
$v_{low}$ are probability distributions, and convergence
$\abs{v_{up}-v_{low}}\to 0$ holds in a hyperreal sense (by an application of
the transfer principle, or more generally, {\L}o\'{s}'s theorem
\cite{Bell1971}). Precisely, the statement would be the following: for every
hyperreal $\eps> 0$ there is a number $K$ such that
$\abs{v_{up}^{(k)}-v_{low}^{(k)}}<\eps$ for all hyperreal $k\geq K$, where
$v_{up}^{(k)},v_{low}^{(k)}$ denote the values in the $k$-th iteration of the
algorithm.

An inspection of Jane Robinson's original convergence proof
\cite{Robinson1951} shows that despite all statements transfer to the
hyperreal setting (consistently with {\L}o\'{s}'s theorem), convergence kicks
in beyond an integer limit that depends on the \emph{largest} entry in the
matrix. More specifically, given $\eps>0$ and letting $t\in\N$ be an
iteration counter, Robinson's proof (cf. lemma 4 in \cite{Robinson1951})
establishes at some point that $\max\vec v(t)-\min\vec u(t)<\eps$ conditional
on
\begin{equation}\label{eqn:convergence-threshold}
t\geq 8\cdot a\cdot t^*/\eps,
\end{equation}
based on an induction argument that delivers $t^*$ through the induction
hypothesis. The crucial point here is the number $a$, which is the absolute
value of the largest element in the payoff-matrix. Now, even if $t$ were (at
some point in the induction) a finite integer,
\eqref{eqn:convergence-threshold} asserts a hyperreal lower bound for $t$. By
our construction, $a$ represents a probability distribution with a diverging
moment-sequence (that is $a$) and thus an infinite hyperreal number.
Consequently, the number of iterations until convergence is a hyperreal and
infinite integer. In other words, convergence of algorithm \eqref{alg:fp},
when implemented in plain form over the hyperreals, cannot converge for a
standard integer for-loop (in line \ref{lbl:fp-for-loop}).

\subsection{Assuring Convergence}\label{sec:convergence-problems}

The problem can be made visible even on a less abstract level, by considering
an intermediate step in the algorithm that has established mixtures $u$ and
$v$ as multimodal distributions. Consider the intermediate values computed by
algorithm \ref{alg:fp}, using the following concrete example for
illustration.
\begin{exa}[Convergence failure of plain FP
implementations]\label{exa:failure-of-fp}
We construct a $2\times 2$-game with payoffs being Epanechnikov density
functions centered around the values
\[
    \vec A=\left(
      \begin{array}{cc}
        2 & 5 \\
        3 & 1 \\
      \end{array}
    \right)
\]
The payoff structure in our game is thus a $2\times 2$-matrix of functions,
\begin{center}
\begin{tabular}{|c|c|}
  \hline
  \includegraphics[scale=1]{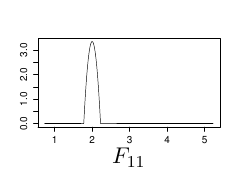} & \includegraphics[scale=1]{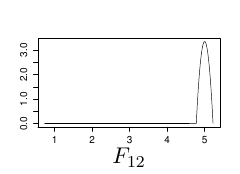} \\\hline
  \includegraphics[scale=1]{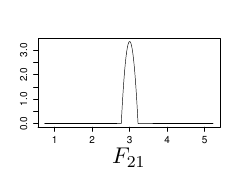} & \includegraphics[scale=1]{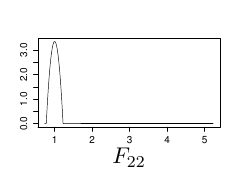}\\
  \hline
\end{tabular}
\end{center}

It is easy to compute the Nash-equilibrium for the exact matrix $\vec A$ as
$v(\vec A)=2.6$, returned by the mixed equilibrium strategies $\vec
p^*=(0.4,0.6)$ and $\vec q^*=(0.8,0.2)$ for both players. Since our setting
shall merely capture our uncertainty about the payoffs, we would thus
naturally expect a somewhat similar result when working on the payoff
\emph{distributions}. Unfortunately, however, algorithm \ref{alg:fp} will not
produce this expected answer, if it is implemented without care.

After a few iterations, the algorithm gets stuck in always choosing the first
row for player 1, since $v_{up}$ is always $\preceq$-preferable over $u^*/k$,
which is immediately obvious from plotting the two distributions:

\begin{center}
\begin{tabular}{ccc}
  \includegraphics[scale=1]{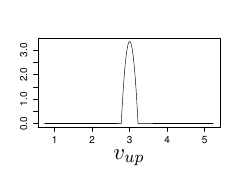} & \raisebox{1.5cm}{$\preceq$} & \includegraphics[scale=1]{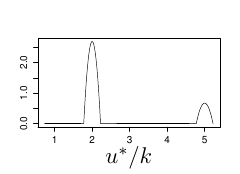} \\
\end{tabular}
\end{center}

Observe that choosing the upper row in the payoff structure adds probability
mass to lower damages, but leaves the tail of the distribution unchanged.
Thus, although the overall damage accumulates, this effect is not noticeable
by the $\preceq$-relation. Consequently, the algorithm will
counterintuitively come to the conclusion that $\vec p=(0,1)$ is a pure
equilibrium, which is definitely not meaningful.
\end{exa}

Let us now make the intuition developed in example \ref{exa:failure-of-fp}
more rigorous to find the solution for the problem.

Recall that the bounds are updated (lines \ref{lbl:vup-update} and
\ref{lbl:vlow-update} in algorithm \ref{alg:fp}) when $u^*/k$ or $v_*/k$ are
preferable over the current bounds. Now, given choices of actions whose
payoff distribution has vanishing support at the right end of the relevant
interval (choices $u^*$ and $v_*$), the pointwise scale and accumulation of
distributions will make $1/k\cdot u^*$ always larger than $v_{up}$, simply
because $v_{up}$ assigns no mass to regions where $u^*$ has a strictly
positive density. This strict positivity is clearly retained upon any scaling
by $1/k$, no matter how large $k$ is. Thus, the algorithm gets stuck with the
same action being chosen once and forever, since adding $\preceq$-smaller
payoff distributions to the long-run average $u^*/k$ will never make this
payoff less than $v_{up}$, and thus never cause an update to this value. The
analogous effect prevents updates to $v_{low}$, and the algorithm will end up
with pure equilibria and thus never converge to a practically meaningful
result.

A closer look at the example, however, also reveals the way out of the
dilemma: the algorithm hangs because the cumulative sum adds up probability
densities pointwise and the comparison is done at the tails of the
distribution only. Thus, any change outside the tail region is not noticed
when the $\preceq$-relation is evaluated, thus never triggering any updates
to $v_{up}$ and $v_{low}$. Obviously, the effect disappears if all
distributions in the payoff structure have full support over the interval of
interest, so that any change to the cumulative payoff is ``noticeable
everywhere'', especially at the tail region of the distribution.

This is easily accomplished by replacing the Epanechnikov-kernels by Gaussian
kernels in the estimator \eqref{eqn:kernel-density-estimator}, and truncating
the kernel density estimates $\hat f$ at a fixed (chosen) point; as
recommended already in section \ref{sec:paradoxical-choices}.

To establish the correctness of this fix, we must assure convergence of
fictitious play under these new payoff structures. To this end, we will
resort to the original version of fictitious play being done over $\R$. In
fact, we will construct a real-valued approximation of the hyperreal
representation of the payoff structure, so that we can work with standard FP
again. The transition from the hyperreal representation back to a real-valued
representation is approximative in the sense that we will represent a
hyperreal value/distribution $F$ by a real value $y(F)$, to be defined later,
in an $\preceq$-order-preserving fashion.

Let $\hat f$ be a Gaussian kernel density estimator for a payoff distribution
$F$ (more precisely, its density function), which is strictly positive on the
entire interval $\Omega=[0,a]$. Since all payoff densities are strictly
positive at $a$, we first observe that for any two densities $f_1,f_2$, we
have
\[
    \hat f_1(a)<\hat f_2(a)\text{ implies } \hat f_1\preceq\hat f_2.
\]
Otherwise, if $f_1(a)=f_2(a)$ (the $>$-case is analogous as before), then the
first order derivative may tip the scale in the sense that under equality of
the densities at $x=a$,
\[
\hat f'_1(a)>\hat f'_2(a)\text{ implies } \hat f_1\preceq\hat f_2.
\]
Upon equality of values and slopes, the second order derivative makes the
decision, since in that case,
\[
\hat f''_1(a)<\hat f''_2(a)\text{ implies } \hat f_1\preceq\hat f_2,
\]
and so on. All these three concrete cases are most easily visualized
graphically (see Figure \ref{fig:value-derivative-comparison} for the first
two cases), and our next step is making this so far intuitive criterion
rigorous. For that matter, recall that the \emph{lexicographical order} on
two (infinite) sequences $\vec a = (a_n)_{n\in\N}$ and $\vec b =
(b_n)_{n\in\N}$. This ordering is denoted as $<_{lex}$ and calls $\vec
a<_{lex} \vec b$ if $a_1<b_1$. If $a_1=b_1$, then $\vec a<_{lex} \vec b$ if
$a_2<b_2$, and so on. For strings (whether finite or not), the
lexicographical order is simply the alphabetical order. The lexicographical
less-or-equal $\leq_{lex}$ order is defined in the obvious way.

\begin{figure}
  \centering
  \subfloat[$f(a)<g(a)\Rightarrow f\preceq g$]{
  \includegraphics[width=0.45\textwidth]{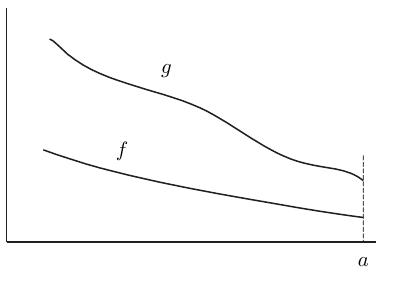}}
  \subfloat[$(f(a)=g(a)\land -f'(a)<-g'(a))\Rightarrow f\preceq g$]{
  \includegraphics[width=0.45\textwidth]{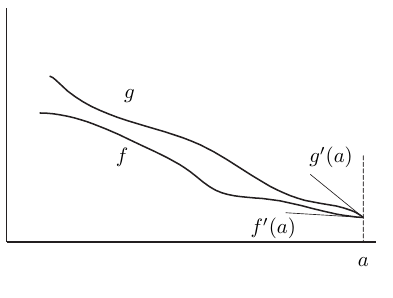}}
  \caption{Comparing tails using function values and derivatives}\label{fig:value-derivative-comparison}
\end{figure}

\begin{lem}\label{lem:derivative-comparisons}
Let $f,g\in C^\infty([0,a])$ for a real value $a>0$ be probability density
functions. If
\[
    ((-1)^k\cdot f^{(k)}(a))_{k\in\N}<_{lex} ((-1)^k\cdot g^{(k)}(a))_{k\in\N},
\]
then $f\preceq g$.
\end{lem}
Thus, it suffices to lexicographically compare the sequence of derivatives
with alternating sign-change to decide which distribution is the preferred
one.

\begin{proof}[of lemma \ref{lem:derivative-comparisons}]
The argument is based on our intuition above, and will look for which
distribution is below the other in a right neighborhood of $a$. To simplify
matters, however, let us ``mirror'' the functions around the vertical line at
$x=a$ and look for which of $f(x), g(x)$ grows faster when $x$ becomes larger
than $a$, using an induction argument on the derivative order $k$. Clearly,
whichever function grows slower for $x\geq a$ in the mirrored view is the
$\preceq$-preferable one. Furthermore, we may assume $a=0$ without loss of
generality (as this is only a shift along the horizontal line). For $k=0$, we
have $f(0)<g(0)$ clearly implying that $f\preceq g$, since the continuity
implies that the relation holds in an entire neighborhood $[0,\eps)$ for some
$\eps>0$. Thus, the induction start is accomplished.

For the induction step, assume that $f^{(i)}(0)=g^{(i)}(0)$ for all $i<k$,
$f^{(k)}(0)<g^{(k)}(0)$, and that there is some $\eps>0$ so that
$f^{(k)}(x)<g^{(k)}(x)$ is satisfied for all $0\leq x<\eps$. Take any such
$x$ and observe that
\begin{align*}
    0 &> \int_0^x\left(f^{(k)}(t)-g^{(k)}(t)\right)dt =f^{(k-1)}(x)-f^{(k-1)}(0)-\left[g^{(k-1)}(x)-g^{(k-1)}(0)\right]\\
    &=f^{(k-1)}(x)-g^{(k-1)}(x),
\end{align*}
since $f^{(k-1)}(0)=g^{(k-1)}(0)$ by the induction hypothesis. Thus,
$f^{(k-1)}(x)<g^{(k-1)}(x)$, and we can repeat the argument until $k=0$ to
conclude that $f(x)<g(x)$ for all $x\in[0,\eps)$.

For returning to the original problem, we must only revert our so-far
mirrored view by considering $f(-x), g(-x)$ in the above argument. The
derivatives accordingly change into $\frac {d^k}{dx^k}f(-x)=(-1)^k
f^{(k)}(x)$, and the proof is complete.
\end{proof}

Lemma \ref{lem:derivative-comparisons} has a useful corollary, which
establishes the sought order-preserving property of the new sequences.
\begin{cor}\label{cor:derivative-comparisons}
For $f\in C^{\infty}([0,a])$ being a probability distribution supported on
$[0,a]$, define $\vec y(f) := ((-1)^k f^{(k)}(a))_{n\in\N}$. Then, for any
two such distributions $f,g$, we have
\[
    f\preceq g\iff \vec y(f)\leq_{lex} \vec y(g).
\]
\end{cor}
\begin{proof}
Let us consider the case when $\vec y(f)\neq \vec y(g)$ first. If so, then
$f\preceq g$ must necessarily make $\vec y(f)\leq_{lex} \vec y(g)$; for
otherwise, if $\vec y(f)>_{lex}\vec y(g)$, then lemma
\ref{lem:derivative-comparisons} would imply $g\prec f$, a contradiction.

Conversely, if $\vec y(f)<_{lex}\vec y(g)$, then lemma
\ref{lem:derivative-comparisons} implies $f\preceq g$ directly.

When $\vec y(f)=\vec y(g)$, then all derivatives are identical. This also
means that $f,g$ have identical Taylor-series expansions, and are therefore
identical functions. Thus, $f\equiv g$ consequently.
\end{proof}

The benefit of replacing an infinite sequence (a hyperreal number) by another
infinite sequence of derivatives lies in possibility to approximate the
density at its tail in a way that gets more and more accurate, depending on
how many terms we include in the sequence. That is, we can easily resort to a
Taylor-polynomial that approximates the tail-behavior of the density
functions up to any precision we like. Moreover, the Gaussian density has the
particular appeal of having all its distributions computable in closed form
expressions by Hermite-polynomials. The $n$-th such polynomial can be
constructed recursively as
\[
    H_{n+1}(x) := 2x H_n(x)-2nH_{n-1}(x), \quad\text{with }H_0(x)=1, H_1(x)=2x.
\]
These are related to the $n$-th order derivative of the Gaussian density via
the identity
\[
    (-1)^ne^{\frac{x^2}2}\frac{d^n}{dx^n}e^{-\frac{x^2}2}=2^{-\frac n 2}H_n\left(\frac x{\sqrt 2}\right).
\]
In \texttt{R}, the Hermite-polynomials can be computed by the
\texttt{orthopolynom} package \cite{Novomestky2013}.

\subsection{Numerical Aspects}\label{sec:numerics}
Now, to ultimately fix the convergence issue of FP over the hyperreals, we
apply corollary \ref{cor:derivative-comparisons} on the \emph{truncated}
series of moments, or equivalently, to a $t$-th order Taylor-polynomial
fitted to the kernel densities under investigation.

However, two more aspects deserve attention:
\begin{itemize}
  \item any floating point representation of the derivatives is inevitably
      subject of a rounding errors, say the machine precision $\eps_M>0$.
  \item the number of necessary terms in the Taylor-series expansion at
      $x=a$,
        \begin{equation}\label{eqn:taylor-approximation}
        \hat f(x) = \hat f(a) + \sum_{k=1}^\infty \frac{\hat f^{(k)}(a)}{k!}(x-a)^k
        \end{equation}
      may grow too large to be practically feasible (see figure
      \ref{fig:gauss-taylor} for an approximation that includes derivatives
      up to order 20 but is quite inaccurate near the origin $x=0$).
      However, given that we only need a good tail-approximation, the
      quality of the approximation far from the end of the support may be
      practically negligible. In other words, the result of the comparison
      remains somewhat insensitive against approximation errors far from
      the tail.

\begin{figure}
  \centering
  \includegraphics[scale=0.5]{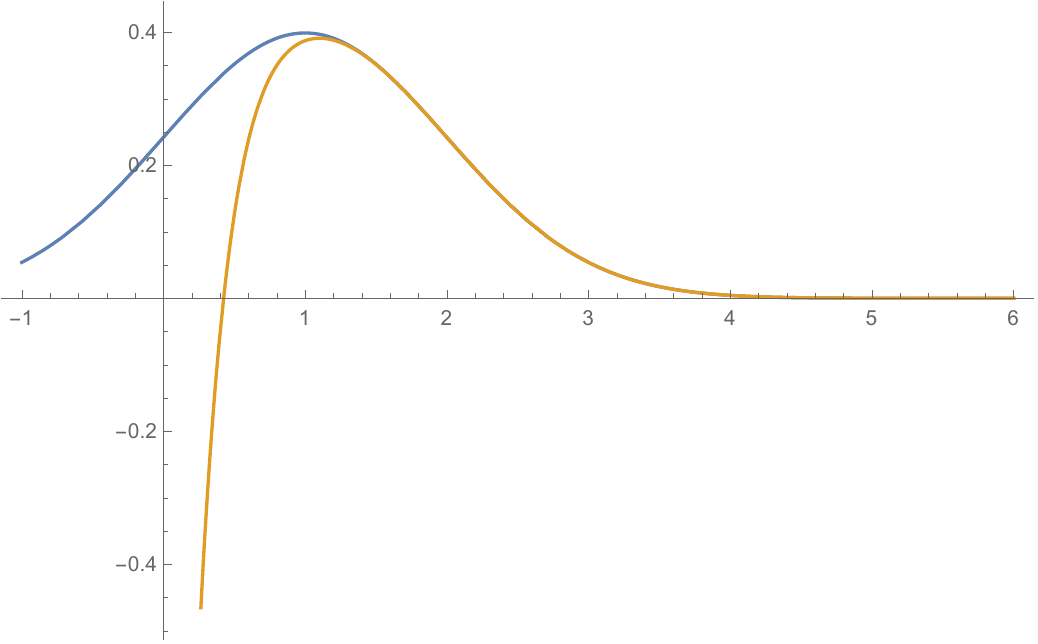}
  \caption{Approximating a $N(1,1)$-density at $x=4$ by a Taylor polynomial of order $20$}\label{fig:gauss-taylor}
\end{figure}

\end{itemize}
The roundoff error can directly be computed from
\eqref{eqn:taylor-approximation}, since the approximate version $\hat
f_{\eps}$ differs from $\hat f$ by at most
\[
    \max_{x\in[0,a]}\abs{\hat f(x)-\hat f_{\eps}(x)}\leq \eps_M + \sum_{k=1}^\infty \frac{\eps_M}{k!}a^k=\eps_M\cdot e^a,
\]
and can thus be kept under control by keeping the roundoff error small,
according to the (known) value $a$ (where all distributions are jointly
truncated).

Given that all derivatives can (numerically) be represented with $\ell$-bit
floating point words in a suitable excess representation\footnote{There will
be no need to implement such a representation in reality, as we require it
only for our theoretical arguments; however, an excess representation makes
the numerical order identical to the lexicographic order, provided that the
sign bits are chosen suitably (i.e., a 0 indicating a negative number, and 1
marking a positive number).}, we can cast the (now finite) sequence $\vec
y(f) := ((-1)^k f^{(k)}(a))_{n\in\N}$ from corollary
\ref{cor:derivative-comparisons} into a real number by concatenating the bit
representations of the numbers into a new value
\begin{equation}\label{eqn:floating-point-approximation}
    y(f) := f(a)\|-f'(a)\|f''(a)\|-f'''(a)\|\ldots\|(-1)^t f^{(t)}(a),
\end{equation}
where $t$ is the number of derivatives used, $\|$ denotes the bitstring
concatenation, and all entries are rounded floating point values (in excess
representation). The representation of a density function $f$ by a real
number of the form \eqref{eqn:floating-point-approximation} has its numerical
order (say between $y(f), y(g)$) being identical to the lexicographical order
on the respective sequences/strings (as the highest order bits determine the
order, and upon equality, the lower order bits come into play, etc.). Thus --
in theory -- we can safely use $y(f)$ as an approximate and order-preserving
replacement for the density function $f$.

More importantly, the representation \eqref{eqn:floating-point-approximation}
is somewhat ``retained'' for mixed distributions like
\eqref{eqn:fp-intermediate-mix}, as algorithm \ref{alg:fp} constructs: assume
that in \eqref{eqn:floating-point-approximation}, all entries $(-1)^i
f^{(i)}(a)$ range within $0\leq \abs{f^{(i)}(a)}\leq M$ for some constant
$M$. This constant must exist, as it is merely the maximum among all suprema
that the derivatives up to order $t$ attain on the (compact) interval $[0,a]$
(recall that our Gaussian kernel density estimator is continuously
differentiable). Since the convex combination \eqref{eqn:fp-intermediate-mix}
of distributions by linearity of the differentiation amounts to a convex
combination of the respective derivatives, the results must all stay within
the same numeric range $[-M,+M]$, thus adding up and averaging the payoffs
does not create a ``carry'' from the part $(-1)^i f^{(i)}(a)$ into the left
neighboring part $(-1)^{i-1} f^{(i-1)}(a)$ within the bitpattern of
expression \eqref{eqn:floating-point-approximation}. Thus, the representation
remains valid also for the mixes \eqref{eqn:fp-intermediate-mix} maintained
over the execution of FP.

Matters of implementation can be further simplified: Since the numeric order
of the number $y(f)$ given by \eqref{eqn:floating-point-approximation} is the
same as the lexicographic order on the truncated sequence $\widetilde{\vec
y}(f)=((-1)^i f^{(i)}(a))_{i=0}^t$, we can more conveniently work in practice
with the vector representation $\vec y(f)$, although the formal convergence
argument for the fictitious play algorithm is established in the standard
setting over the reals, a we can use $y(f)$ in theory.

In other words, we use the real value $y(f)$ and $\leq$-relation on the
payoffs in theory, while working with (the truncated) vector-version
$\widetilde{\vec y}(f)$ and $\leq_{lex}$ relation in practice, when running
FP. Convergence then follows from Robinson's results \cite{Robinson1951}.

Having assured convergence of fictitious play by switching back to numeric
approximate representations $y(f)$ of distribution functions $f$, it remains
to establish the \emph{correctness} of the approximate equilibrium returned
by the so-modified algorithm. This easily follows from what we have obtained
so far:

\begin{quote}
\emph{Informally}: A strategy profile is an ``approximate'' equilibrium in
the game with hyperreal (distributional) payoffs $F_{ij}$, if and only if
it is an equilibrium in the game over the reals, using $y(f_{ij})\in\R$ as
given by \eqref{eqn:floating-point-approximation} in its payoff structure.
\end{quote}

\noindent The following reasoning assumes a minimizing first player:

``$\Rightarrow$'': Let $\vec p^*=(\vec p_1^*,\ldots, \vec p_n^*)$ be a
Nash-equilibrium (in mixed strategies) in the $n$-person game with
distribution-valued payoff distribution functions $F_i(\vec p_i,\vec
p_{-i}^*)$. To ease notation, let us write $f_i(\vec p_i,\vec p_{-i}^*)$ to
mean the respective density function. Since $\vec p_i^*$ is an equilibrium
strategy for player $i$, we have $\forall i\forall \vec p: F_i(\vec p_i^*,
\vec p_{-i}^*)\preceq F_i(\vec p,\vec p_{-i}^*)$. Since both sides are
themselves distribution functions, corollary
\ref{cor:derivative-comparisons}, tells that their approximate
representatives satisfy $y(f_i(\vec p_i^*, \vec p_{-i}^*))\leq y(f_i(\vec
p,\vec p_{-i}^*))$.

``$\Leftarrow$'': Conversely, if $y(f_i(\vec p_i^*, \vec p_{-i}^*))\leq
y(f_i(\vec p,\vec p_{-i}))$, then we may distinguish two cases:
\begin{enumerate}
  \item An identity $y(f_i(\vec p_i^*, \vec p_{-i}^*))=y(f_i(\vec p,\vec
      p_{-i}^*))$ implies that the respective Taylor-poly\-nomial
      (approximations) are also identical, and thus the density $f_i(\vec
      p_i^*, \vec p_{-i}^*)$ is identical to the density $f_i(\vec p,\vec
      p_{-i}^*)$. Hence, $F_i(\vec p_i^*, \vec p_{-i}^*)\preceq F_i(\vec
p,\vec p_{-i}^*)$ in particular.
  \item Strict inequality $y(f_i(\vec p_i^*, \vec p_{-i}^*))< y(f_i(\vec
      p,\vec p_{-i}^*))$, however, implies $F_i(\vec p_i^*, \vec
      p_{-i}^*)\preceq F_i(\vec p,\vec p_{-i}^*)$ by corollary
      \ref{cor:derivative-comparisons}, as before.
\end{enumerate}

We may thus conclude that an approximate Nash-equilibrium can be computed by
fictitious play through a double approximation:
\begin{enumerate}
  \item We approximate the unknown (real) payoff distribution by a Gaussian
      kernel density estimator and truncate all the payoff estimators at
      the same value $a$.
  \item We fit a Taylor-polynomial of order $t$ at $x=a$ to enable the
      decision of the $\prec$-relation on the approximation by use the the
      vector of derivatives $(-1)^i\cdot f^{(i)}(a)$ for $i=0,1,\ldots, t$.
\end{enumerate}

\begin{rem}
We stress that the above arguments \emph{are not} a theorem so far, as the
quality of approximation of the resulting Nash equilibrium compared to the
theoretically exact result remains undetermined. Indeed, both, the kernel
density estimate and the Taylor polynomial get more and more accurate the
more points we sample (for the kernel density) or the more derivatives we
take into account (for the Taylor approximation). A precise error estimate,
however, is involved and beyond the scope of this current article (subject of
future investigations). Hence, we refrain from calling the last findings a
``theorem'' or similar here. Still, it must be emphasized that this double
approximation can be made arbitrarily accurate, where $\preceq$ is decided.
\end{rem}

\paragraph{Discrete Distributions}
To handle discrete distributions, we can ``abuse'' the kernel density
estimators in a straightforward fashion, as the effect of incorporating more
and more samples at few fixed positions leads to a set of Dirac-masses
accumulating at exactly the points of the support. Somewhat roughly speaking,
this causes no immediate harm, since the likelihood $\Pr(X\leq
k)=\sum_{i=0}^k\Pr(X=i)$ can quite well be approximated by integrating
Gaussian densities from $0$ to $k$, provided that the bell curve has its mass
centred slightly left of the integer point (say, half-spaced between the
integers), so that almost all the probability mass concentrates between two
integers (this is somewhat similar to an empirical histogram, with the
exception that it creates Dirac-masses instead of bars in the diagram).

The so-far sketched method thus remains applicable.

\subsection{Remaining Possible Convergence
Issues}\label{sec:alis-convergence-issues} Although fictitious play generally
will converge, it may not do so in reasonable time, as experimental findings
showed. Let us view the computation from a different angle: by Lemma
\ref{lem:derivative-comparisons} and also Theorem 3 in
\cite{rass_decisions_2016}, we can replace a loss distribution $F$ by some
(possibly infinite) vector $\vec f=(f_1,f_2, \ldots)$ so that $F\preceq G$ if
and only if $\vec f\leq_{lex} \vec g$. Let us in the following consider these
vectors as representatives of our loss distributions, and look at their
coordinates. Fictitious play, in the initial stage, is merely a game being
played on the first coordinate $f_1$ of the payoff structures until it
converges. By then, it breaks the tie and moves on to the second coordinate
$f_2$, to continue the game from there, further adapting the strategies for
both players. This, however, may invalidate the optimum found this far, thus
causing FP to return to the last coordinate ($f_1$) again to fix the issue.
This means that, given any accuracy by which FP would stop iterating, we
would end up with cycling between the first and the second coordinate, until
both have been optimized. Not until these two converged, we would move
further to the third coordinate in the lexicographic order, and so on. This
can practically mean that FP may still take an infeasibly large number of
iterations, unless the accuracy is set sufficiently coarse to accept ties
within a certain proximity of the actual optimum.

The view of lexicographic optimization as a stack of games, however, also
leads to an exact method to compute the optimum, using linear programming in
the next section.

\subsection{Exact Solution by Linear Programming}\label{sec:lp}
It is a standard exercise in game theory to convert the saddle point problem
\[
    \min_{\vec x}\max_{\vec y}\vec x^T\vec A\vec y
\]
with $\vec A\in\R^{n\times m}$ into the linear program
\begin{alignat}{2}
    \text{maximize}\quad & v \label{eqn:lp}\\
    \text{subject to}\quad  & v\leq\vec x^T\vec A\vec e_j &\quad\text{for all }j=1,2,\ldots,m;\nonumber\\
    & x_1+x_2+\ldots x_n =1;\nonumber\\
    & x_i \geq 0&\quad \text{for all }i=1,2,\ldots n.\nonumber
\end{alignat}
using $\vec e_i$ as the $i$-th unit vector in $\R^n$.

Now, taking into account that $\preceq$-optimization is actually a matter of
optimization respecting the lexicographic order, we actually have a stack of
games with payoff matrices $\vec A_1,\vec A_2,\vec A_3,\ldots$, in which
$\vec A_i$ is the matrix game using only the $i$-th coordinate in the
vector-representation of the payoff distributions. For the example of Lemma
\ref{cor:derivative-comparisons}, we would have $\vec A_k$ with entries
$(-1)^k\cdot f^{(k)}_{ij}(a)$ for the strategy pair $(i,j)\in PS_1\times
PS_2$; in the discrete case, the entries would directly be the probability
masses for the respective category.

Like for FP, we can then exactly compute an MGSS by solving \eqref{eqn:lp} on
the game $\vec A_1$ containing only the last coordinates in the payoff
vectors. This delivers the saddle-point value $v_1$, and is the same is if FP
would have taken infinitely many iterations to convergence.

Now, to avoid the tie breaking issue described in Section
\ref{sec:alis-convergence-issues}, we move on to optimize the next
coordinate, but \emph{preserving} what we have achieved so far. That is, the
next linear program will use the matrix $\vec A_2$, but demand any feasible
solution $\vec x$, and hence also the optimum, to still satisfy the further
constraint
\begin{equation}\label{eqn:constraint}
    \vec x^T\cdot \vec A_1 \leq v_1,
\end{equation}
where the $\leq$ is w.r.t. to all rows in the vector-matrix product. This
prevents invalidating our previous optimum if we go on minimizing the mass on
the butlast loss category, or second coordinate in the vector representation.

From that point onwards, the process continues likewise, i.e., if the payoff
distributions $F_{ij}$ are represented by vectors $\vec f_{ij}\in\R^k$, we
repeat the following steps, starting from $i=1$:
\begin{enumerate}
  \item Construct the payoff structure $\vec A_i$ by putting the $i$-th
      coordinate of $\vec f_{ij}$ into the $ij$-position of the matrix
      $\vec A_i$.
  \item Construct the linear program \eqref{eqn:lp} using $\vec A_i$.

        If $i>1$, add the constraints $\vec x\cdot \vec A_j\leq v_j$ for
  $j=1,\ldots,i-1$ to \eqref{eqn:lp}.

  \item Solve the resulting optimization problem to obtain $v_i$. Put
      $i\gets i+1$ and return to step 1 until all coordinates have been
      processed.
\end{enumerate}
The vector $\vec x^*$ coming out of the final optimization is the sought
multi-goal security strategy. The same procedure, only with the respective
dual linear programs with the respectively changed constraint
\eqref{eqn:constraint} applies for each (hypothetical) opponent to lex-order
optimize its own goal in the zero-sum game against the defending player 0.

\subsection{Implementation Hints for Algorithm \ref{alg:fp} --
Overview}\label{sec:generalized-fp-overview} As we noticed above, the
implementation of algorithm \ref{alg:fp} \emph{cannot} be done plainly and
using the decision method sketched in section
\ref{sec:distribution-comparisons}. Instead, the following adaptations
require a careful implementation:
\begin{enumerate}
  \item Use Gaussian kernel density estimators to approximate the actual
      payoff distribution functions.
  \item Fix a common cutoff point $a$ and truncate all distributions in the
      entire payoff matrix at the value $a$.
  \item Fix an order $t$ and compute Taylor-polynomial approximations for
      all payoff distributions at the point $x=a$. To represent an
      (approximate) distribution in the payoff matrix, use the vector of
      derivatives with alternating signs (as defined before).
  \item Implement algorithm \ref{alg:fp} as it stands, using corollary
      \ref{cor:derivative-comparisons} to perform the $\preceq$-comparisons
      (on the approximate representatives of the true payoff structure).
      The pointwise additions of densities herein can be applied just as
      they are to the vectors of derivatives instead (thanks to the
      linearity of the differentiation operators).
\end{enumerate}

\subsection{The case of $d>1$ security goals}\label{sec:general-mgss-fp}
Our presentation here is directly adapted from \cite{Rass2014}: given a
two-player multi-objective $\Gamma$ and its auxiliary game
$\overline{\Gamma}$ as specified in \cite{Rass2015b}, we first need to cast
$\overline{\Gamma}$ into zero-sum form to make FP converge. To this end,
recall that player $1$ in $\Gamma$, who is player $0$ in $\overline{\Gamma}$,
has $d$ goals to optimize, each of which is represented as another opponent
in the auxiliary game $\overline{\Gamma}$. We define the payoffs in a
\emph{compound game} (``one-against-all'', where the ``one'' is the defender,
against ``all'' (hypothetical) adversaries) from the payoffs in
$\overline{\Gamma}$, while making the scalar payoffs vector-valued to achieve
the zero-sum property:
\begin{itemize}
  \item player 0:
  \begin{eqnarray*}
  &&\overline{\vec u}_0: PS_1\times \prod_{i=1}^d PS_2\To\R^d,\\
      &&\vec u_0(s_0,\ldots,s_d):=(u_1^{(1)}(s_0,s_1),u_1^{(2)}(s_0,s_2),\ldots,u_1^{(d)}(s_0,s_d))
  \end{eqnarray*}
  \item $i$-th opponent for $i=1,2,\ldots,d$:
  \begin{equation}\label{eqn:opponent-payoffs}
    \overline{\vec u}_i = (0,0,\ldots,0,-u_1^{(i)},0,\ldots,0).
  \end{equation}
\end{itemize}
Obviously, the ``vectorization'' of the opponents payoffs leaves the set of
equilibria unchanged.

\subsubsection{Prioritization of Security Goals}
To numerically compute (one of) the MGSS, \cite{Lozovanu2005} prescribes to
scalarize the multiobjective game into a single-objective game. Indeed, this
scalarization is primarily for technical reasons and can be made arbitrary
(under suitable constraints). However, it this degree of freedom in the
choice on the scalarization allows us to \emph{prioritize} the security goals
of player 0: To each of his $d$ goals, we assign a \emph{weight}
$\alpha_{01},\ldots,\alpha_{0d}>0$ subject to the condition that
$\sum_{i=1}^d \alpha_{0i}=1$.

With these weights, the payoffs after scalarization are:
\begin{itemize}
  \item for player $0$: $f_0 = \alpha_{01}\overline{\vec u}_1 +
      \alpha_{02}\overline{\vec u}_2 + \cdots + \alpha_{0d}\overline{\vec
      u}_d$,
  \item for the $i$-th opponent, where $i=1,2,\ldots,d$
    \begin{eqnarray}
          f_i &=& \alpha_{01}\cdot
      0+\alpha_{02}\cdot 0+\cdots + \alpha_{0,i-1}\cdot 0 +
      \alpha_{0i}\cdot (-u_1^{(i)}) + \alpha_{0,i+1}\cdot 0 + \alpha_{0d}\cdot 0\nonumber\\
      &=&-\alpha_{0i}\cdot u_1^{(i)}\label{eqn:opponent-payoffs-scalarized}
    \end{eqnarray}
\end{itemize}
Concluding the transformation, we obtain a scalar compound game
\begin{equation}\label{eqn:scalarized-compound-game}
\overline{\Gamma}_{sc} = (\set{0,1,\ldots,d}, \{PS_1,\underbrace{PS_2,\ldots,PS_2}_{d\mbox{ \scriptsize{times}}}\},
\set{f_0,\ldots, f_d})
\end{equation}
from the original two-person multiobjective $\Gamma$ with payoffs
$u_1^{(1)},\ldots,u_1^{(d)}$ that can directly be plugged into expressions
\eqref{eqn:opponent-payoffs} and \eqref{eqn:opponent-payoffs-scalarized}.

Towards a numerical computation of equilibria in $\overline{\Gamma}_{sc}$, we
need yet another transformation due to \cite{Sela1999}: for the moment, let
us consider a general compound game $\Gamma_c$ as a collection of $d$
two-person games $\Gamma_c^{(1)},\ldots,\Gamma_c^{(d)}$, each of which is
played independently between player 0 and one of its $d$ opponents. With
$\Gamma_c$, we associate a two-person game $\Gamma_{cr}$ that we call the
\emph{reduced game}. The strategy sets and payoffs of player 0 in
$\Gamma_{cr}$ are the same as in $\Gamma_c$. Player 2's payoff in the reduced
game is given as the \emph{sum} of payoffs of all opponents of player 0 in
the compound game $\Gamma_c$.

\begin{lem}[\cite{Sela1999}]
A fictitious play process approaches equilibrium in a compound game
$\Gamma_c$, if and only if it approaches equilibrium in its reduced game
$\Gamma_{cr}$.
\end{lem}

So, it suffices to consider the reduced game $\overline{\Gamma}_{scr}$
belonging to $\overline{\Gamma}_{sc}$. It is a trivial matter to verify the
following fact (by substitution).
\begin{lem} The reduced game $\overline{\Gamma}_{scr}$ of the scalarized compound game
$\overline{\Gamma}_{sc}$ defined by \eqref{eqn:scalarized-compound-game} is
zero-sum.
\end{lem}
So by the convergence of FP in any zero-sum game (theorem
\ref{thm:standard-fp-convergence}), we obtain the final correctness result on
FP when being applied to our $(d+1)$-person game $\overline{\Gamma}_{sc}$
that will deliver the sought MGSS:

\begin{thm}\label{thm:MGSS-FP} The scalarized compound game
$\overline{\Gamma}_{sc}$ defined by \eqref{eqn:scalarized-compound-game} has
the fictitious play property.
\end{thm}

\subsubsection{The Full Procedure}
Theorem \ref{thm:MGSS-FP} induces the following procedure to compute MGSS
according to definition \ref{def:MGSS}:

Given a two-player multi-goal $\Gamma$ with $d$ payoffs
$u_1^{(1)},\ldots,u_1^{(d)}$ for player 1 (and possibly unknown payoffs for
player 2), we obtain an equilibrium in an MGSS along the following steps:
\begin{enumerate}
  \item Assign strictly positive weights $\alpha_{01},\ldots,\alpha_{0d}$,
      satisfying $\sum_{i=1}^d \alpha_{0i}=1$, to each goal, and set up the
      scalarized auxiliary compound game $\overline{\Gamma}_{sc}$ by virtue
      of expressions \eqref{eqn:opponent-payoffs},
      \eqref{eqn:opponent-payoffs-scalarized} and
      \eqref{eqn:scalarized-compound-game}.
  \item Run fictitious play (algorithm \ref{alg:fp}) in
      $\overline{\Gamma}_{sc}$, stopping when the desirable precision of
      the equilibrium approximation is reached. Preferably for a
      non-approximate solution, solve a sequence of linear programs as
      described in Section \ref{sec:lp}.
  \item The result vector $\vec x$ is directly the sought multi-criteria
      security strategy, whose assurances are given by the respective
      expected payoffs of the opponents (equation
      \eqref{eqn:game-outcome-distribution} in connection with the output
      $\vec y$ of the algorithm).
\end{enumerate}

A few remarks about this method appear in order:
\begin{itemize}
  \item The reduced scalarized game $\Gamma_{scr}$ is a purely theoretical
      vehicle to establish convergence of FP in the scalarized (formerly
      multiobjective) multiplayer game $\overline{\Gamma}_{scr}$. Thus, FP
      is to be executed on $\overline{\Gamma}_{scr}$ by copying and
      adapting lines \ref{lbl:start-player2} until \ref{lbl:end-player2}
      (incl.) in algorithm \ref{alg:fp} for each hypothetical adversary.
  \item Convergence of the game is guaranteed only on the Taylor-polynomial
      approximations for the densities $\vec y(f_{ij})$, but not on the
      true densities $f_{ij}$ as such (for these, the arguments about
      convergence failure discussed around equation
      \eqref{eqn:convergence-threshold} and illustrated in example
      \ref{exa:failure-of-fp} remain intact).
  \item The obtained bounds $v_{up}$ and $v_{low}$ also refer to the
      Taylor-polynomial approximations for the respective bounds on the
      actual value $v$ of the game over the distribution-valued payoffs.
      Thus, the actual assurance limits (distributions $V_i$ given by
      equation \eqref{eqn:ds-req0} in definition \ref{def:MGSS}) must be
      computed from the resulting mixed strategies (as specified in
      algorithm \ref{alg:fp}, line \ref{lbl:fp-finish}, i.e., the formula
      should be implemented for all $d$ goals).
  \item FP and likewise linear programming on the multiplayer game does not
      deliver a concrete (Pareto-Nash) equilibrium strategy for the
      (single) physical adversary, but returns worst-case behavior strategy
      options for every of his $d$ goals. Thus, the result may
      pessimistically underestimate what happens in reality (as the true
      adversary is forced to choose a single out of the multiple options,
      thus necessarily deviating from some of the $d$ equilibrium
      strategies).
  \item The assurances obtained also need to be interpreted bearing in mind
      that here neither player follows the other. That is, the Pareto-Nash
      equilibrium would gives an optimal strategy for the attacker under
      the hypothesis that it plays a zero-sum game against the defender as
      if player 1 would have only this particular goal and no other. The
      real defender, however, will care about multiple goals at the same
      time, thus -- by symmetry -- deviating from the zero-sum equilibrium
      strategy that the attacker has in its own single-goal game. This
      means that the assurance obtained in the last step of the above
      procedure are \emph{not} best replies to the optimal defense, but
      rather the worst-case that would be possible if the defender were to
      spend all its resources on this particular goal.

      Asking for a best reply to the defender's optimal multi-criteria
      strategy is a much simpler issue: for the attacker it merely means to
      adapt by picking the $\preceq$-maximum from the vector $(\vec
      x^*)^T\vec A$, when $\vec A$ is the weighted sum of all payoff
      structures for all goals (doing the scalarization), i.e., the payoff
      structure in the scalarized reduced game constructed in step 2 above.
      This is the case of a \emph{leading-defender and following-attacker},
      always giving a pure worst-case attack scenario, whose payoff is then
      (trivially) a bound to the defender's possible damage.
\end{itemize}



\section{Summary and Outlook}
Implementing the theoretical concepts introduced in \cite{Rass2015b} requires
care to avoid a variety of subtle difficulties. While using
Epanechnikov-kernel based density estimates greatly eases matters of plain
payoff distribution $\preceq$-comparisons, such models are not useful at all
for computing multi-goal security strategies. For these, it is practically
advisable to use Gaussian kernels, being truncated at a common point. At
first, this avoids paradoxical results (as outlined in section
\ref{sec:paradoxical-choices}), but furthermore enables the application of
conventional fictitious play. With example \ref{exa:failure-of-fp}
demonstrating the failure of FP on $^*\R$ when implemented plainly, the
algorithm can nevertheless be successfully applied to proper approximations
of the game matrices. In particular, a double-approximation is applied here,
where the first stage approximates the unknown true distribution by a kernel
density estimate, and the second stage is a Taylor-polynomial expansion at
the tails of the distribution. Despite this approximation performing possibly
badly in areas of low damage (cf. figure \ref{fig:gauss-taylor}), it provides
a good account for the tail-dependent comparison that $\preceq$ is based on.
In other words, even though the Taylor-polynomial is not accurate everywhere
on the relevant interval of possible losses, it nevertheless produces correct
$\preceq$-preferences as these depend on the tails of the distribution, where
the Taylor-approximation is indeed quite accurate. The practical benefit lies
in the ability of running conventional fictitious play on the so-approximated
distribution models, after casting the game-matrix of distributions into a
regular game matrix over the reals.

With the algorithmic aspects being covered in this report, companion (follow
up) work will discuss the modeling and treatment of \emph{advanced persistent
threats}, and how to apply our algorithms to distributions with fat, heavy or
long tails (all of which can be fed back into our algorithms after
truncation, to mention one quick-and-easy solution).

\section*{Acknowledgment}
This work was supported by the European Commission's Project No. 608090,
HyRiM (Hybrid Risk Management for Utility Networks) under the 7th Framework
Programme (FP7-SEC-2013-1). The author thanks Sandra K\"{o}nig from AIT very
much for valuable discussions and for spotting some errors, and also for
providing the demonstration implementation of algorithm \ref{alg:fp}.

The author is also deeply indebted to Ali Alshawish and Vincent Bürgin for
having provided deep insights about the convergence of fictitious play along
intensive experiments. The convergence issue described in Section
\ref{sec:alis-convergence-issues} has been reported by them, and here
repeated in the authors own words. The new computational method using linear
programming was invented in response to this issue.

\bibliographystyle{plain}

\end{document}